\documentclass[11pt]{article}
\usepackage[utf8]{inputenc}
\usepackage{graphicx}
\usepackage[margin=1in]{geometry}
\usepackage{amsmath,amssymb,amsthm,verbatim}
\usepackage[ruled,vlined]{algorithm2e}
\usepackage{algorithmic}
\usepackage[colorlinks=true, linkcolor=black, citecolor=black]{hyperref}
\usepackage{cleveref} 
\usepackage{thmtools, thm-restate}

\usepackage[dvipsnames]{xcolor}
\usepackage{tikz,pgfplots}
\pgfplotsset{compat=1.18} 
\usetikzlibrary{patterns}
\usetikzlibrary{arrows.meta}
\usepackage{caption, setspace, subcaption}

\usepackage{lineno}

\newtheorem{lemma}{Lemma}

\newtheorem{fact}{Fact}
\newtheorem{observation}{Observation}
\newtheorem{invariant}{Invariant}
\theoremstyle{definition}
\newtheorem{definition}{Definition}

\newcommand{\PONE}{Weighted $k$-Server -- Service Pattern Construction}
\newcommand{\WKS}{weighted $k$-server}
\newcommand{\PTWO}{Weighted $k$-server -- Revealed Service Pattern}
\newcommand{\pone}{W$k$S-SPC}

\newcommand{\ptwo}{W$k$S-RSP}

\newcommand{\mone}{forced movement}
\newcommand{\mtwo}{unforced movement}
\SetKwFor{RepTimes}{repeat}{times}{end}

\renewcommand{\AA}{\ensuremath{\mathcal A}\xspace}

\newcommand{\II}{\ensuremath{\mathcal I}\xspace}
\newcommand{\JJ}{\ensuremath{\mathcal J}\xspace}

\ifdefined\DEBUG

\newcommand{\aay}[1]{\textcolor{Green}{#1}}
  \def\rem#1{{\marginpar{\raggedright\scriptsize #1}}}
  \newcommand{\ashr}[1]{\rem{\textcolor{Red}{$\bullet$ #1}}}
  \newcommand{\aayr}[1]{\rem{\textcolor{Green}{$\bullet$ #1}}}
  \newcommand{\nikr}[1]{\rem{\textcolor{Blue}{$\bullet$ #1}}}
\else
  \newcommand{\aay}[1]{#1}

  \newcommand{\aayr}[1]{}
  \newcommand{\ashr}[1]{}
  \newcommand{\nikr}[1]{}
\fi

\DeclareMathOperator{\poly}{poly}

\title{A Decomposition Approach to the Weighted \texorpdfstring{$k$}{k}-server Problem} 

\author{
  Nikhil Ayyadevara\thanks{University of Michigan, Ann Arbor} \and
  Ashish Chiplunkar\thanks{Indian Institute of Technology, Delhi, India, \url{https://www.cse.iitd.ac.in/\~ashishc/}}
  \and
  Amatya Sharma\thanks{University of Michigan, Ann Arbor, \url{https://aaysharma.github.io/}}
}

\date{}

\begin{document}
\maketitle

\begin{abstract}
A natural variant of the classical online $k$-server problem is the \textit{weighted $k$-server problem}, where the cost of moving a server is its weight times the distance through which it moves. Despite its apparent simplicity, the weighted $k$-server problem is extremely poorly understood. Specifically, even on uniform metric spaces, finding the optimum competitive ratio of randomized algorithms remains an open problem -- the best upper bound known is $2^{2^{k+O(1)}}$ due to a deterministic algorithm (Bansal et al., 2018), and the best lower bound known is $\Omega(2^k)$ (Ayyadevara and Chiplunkar, 2021).

With the aim of closing this exponential gap between the upper and lower bounds, we propose a decomposition approach for designing a randomized algorithm for weighted $k$-server on uniform metrics. Our first contribution includes two relaxed versions of the problem and a technique to obtain an algorithm for weighted $k$-server from algorithms for the two relaxed versions. Specifically, we prove that if there exists an $\alpha_1$-competitive algorithm for one version (which we call \textit{\PONE}) and there exists an $\alpha_2$-competitive algorithm for the other version (which we call \textit{\PTWO}), then there exists an $(\alpha_1\alpha_2)$-competitive algorithm for weighted $k$-server on uniform metric spaces. Our second contribution is a $2^{O(k^2)}$-competitive randomized algorithm for \PTWO. As a consequence, the task of designing a $2^{\poly(k)}$-competitive randomized algorithm for weighted $k$-server on uniform metrics reduces to designing a $2^{\poly(k)}$-competitive randomized algorithm for \PONE. Finally, we also prove that the $\Omega(2^k)$ lower bound for weighted $k$-server, in fact, holds for \PTWO.
\end{abstract}

\pagenumbering{arabic}

\section{Introduction}\label{sec:intro}

The $k$-server problem proposed by Manasse et al.~\cite{ManasseMS_STOC88} is a fundamental problem in online computation, and it has been actively studied for over three decades. In this problem, we are given a metric space $M$ and $k$ identical servers $s_1,\ldots,s_k$ located at points of $M$. In every round, a point of $M$ is requested, and an online algorithm serves the request by moving (at least) one server to the requested point. The objective is to minimize the total distance traversed by all $k$ servers. 

Like several other online problems, the performance of algorithms for the $k$-server problem is measured using the framework of competitive analysis introduced by Sleator and Tarjan~\cite{SleatorT85}. An online algorithm for a minimization problem is said to be $\alpha$-competitive if, on every input, the ratio of the algorithm's (expected) cost to the cost of the optimal solution is at most $\alpha$, possibly modulo an additive constant independent of the online input. In the deterministic setup, Manasse et al.~\cite{ManasseMS_STOC88} showed that no $k$-server algorithm can be better than $k$-competitive on any metric space with more than $k$ points. In their breakthrough result, Koutsoupias and Papadimitriou~\cite{KoutsoupiasP_JACM95} gave the best known deterministic algorithm that is $(2k-1)$-competitive on every metric space, famously known as the Work Function Algorithm (WFA). In the setup of randomized algorithms, it is conjectured that the competitive ratio of $k$-server is \aay{$O(\poly(\log k))$}, and this remains unsolved. Very recently, refuting the so-called \textit{randomized $k$-server conjecture}, Bubeck, Coester, and Rabani~\cite{bubeck2022randomized} exhibited a family of metric spaces on which the randomized competitive ratio of the $k$-server problem is $\Omega(\log^2 k)$.

The $k$-server problem is a generalization of the online paging problem. The paging problem concerns maintaining in a ``fast'' memory a subset of $k$ pages out of the $n$ pages in a ``slow'' memory. In each round, one of $n$ ($\gg k$) pages is requested, and it must replace some page in the fast memory, unless it is already in the fast memory.
The objective is to minimize the number of page replacements.
The paging problem is exactly the $k$-server under the uniform metric on the set of pages. The paging problem has been well studied, and several deterministic algorithms like Least Recently Used (LRU), First In First Out (FIFO), etc.\ are \aay{known to be} $k$-competitive~\cite{SleatorT85}. The randomized algorithm by Achlioptas et al.~\cite{AchlioptasCN_TCS00} is known to be $H(k)$-competitive, matching the lower bound by Fiat et al.~\cite{FiatKLMSY_JAlg91}. Here $H(k)=1+1/2+\cdots+1/k=\Theta(\log k)$.

\subsection{Weighted \texorpdfstring{$k$}{k}-server}

The weighted $k$-server problem, first defined by Newberg~\cite{Newberg91}, is a natural generalization of the $k$-server problem. In the weighted $k$-server problem, the servers are distinguishable: the $i$'th server has weight $w_i$, where $w_1\leq\cdots\leq w_k$. The cost incurred in moving a server is its weight times the distance it travels. The objective is to minimize the total weighted distance moved by all $k$ servers. It is easy to see that an $\alpha$-competitive algorithm for the (unweighted) $k$-server problem has a competitive ratio of at most $\alpha\cdot w_k/w_1$ for the weighted $k$-server problem. However, this bound can be arbitrarily bad as $w_k/w_1$ \aay{is} unbounded. So, the challenge is to establish weight-independent bounds on the competitive ratio of the weighted $k$-server problem. Surprisingly, this simple introduction of weights makes this problem incredibly difficult, and a weight-independent upper bound on the competitive ratio for an arbitrary metric is only known for the case when $k\leq2$~\cite{Sitters_SIAMJC14}. 

Owing to its difficulty on general metric spaces, it is natural to completely understand the weighted $k$-server problem on the simplest class of metric spaces, the uniform metric spaces first. Uniform metric spaces are the ones in which every pair of points is separated by a unit distance.
The objective of this problem translates to minimizing the weighted sum of the number of movements of each server, and thus, this problem is equivalent to paging with the cost of a page replacement dependent on the cache slot it is stored in\footnote{Note that this problem is different from weighted paging~\cite{Young_Algorithmica02}, where the weights are on the pages (points in the metric space) instead on the cache slots (servers). In fact, weighted paging is equivalent to unweighted $k$-server on star metrics.}. 
In their seminal work, Fiat and Ricklin~\cite{FiatR_TCS94} gave a deterministic algorithm for the weighted $k$-server problem on uniform metrics with a competitive ratio doubly exponential in $k$, which was later improved by Bansal et al.~\cite{BansalEKN_SODA18} to $2^{2^{k+2}}$. 
This doubly exponential behavior of the competitive ratio was proven tight by Bansal et al.~\cite{BansalEK_FOCS17} when they showed that the deterministic competitive ratio is no less than $2^{2^{k-4}}$. 

In the randomized setup, the only known algorithm which uses randomization in a non-trivial manner is the memoryless algorithm of Chiplunkar and Vishwanathan~\cite{ChiplunkarV_TAlg20}, which has a competitive ratio of 
$1.6^{2^k}$. Chiplunkar and Vishwanathan also showed that this ratio is tight for the class of randomized memoryless algorithms. Recently, Ayyadevara and Chiplunkar~\cite{AyyadevaraC_ESA21} showed that no randomized algorithm (memoryless or otherwise) can achieve a competitive ratio better than $\Omega(2^k)$. Closing the exponential gap between the $1.6^{2^k}$ upper bound and the $\Omega(2^k)$ lower bound on the randomized competitive ratio is still an open problem.

Very recently, Gupta et al.~\cite{gupta2023efficient} studied the weighted $k$-server problem in the offline and resource augmentation settings, showing the first hardness of approximation result for polynomial-time algorithms.


\subsection{Our Contributions}

Throughout this paper, we focus on the weighted $k$-server problem on uniform metrics, and we avoid mentioning the metric space henceforward. Considering the fact that the competitive ratio of a server problem is typically exponentially better in the randomized setting than the deterministic setting, it is reasonable to conjecture that there exists a randomized $2^{\poly(k)}$-competitive randomized algorithm for weighted $k$-server.
In this paper, we propose a way of designing such an algorithm using our key idea of decomposing the weighted $k$-server problem into the following two relaxed versions. 

\paragraph*{\PONE\ (\pone).} The input is the same as the weighted $k$-server. The difference is that, in response to each request, the algorithm must only commit to the movement of some subset of servers, without specifying where those servers move to. However, it is required that there exists some solution to the given instance that agrees with the algorithm's server movements. Note that the algorithm could potentially benefit from not being lazy, that is, by moving 
more than one server at the same time.
The (expected) cost of the algorithm is, as defined earlier, the weighted sum of the number of movements of each server (recall that we are working on a uniform metric space so the distance between every pair of points is one unit). An algorithm is said to be $\alpha$-competitive if the (expected) cost of its solution is at most $\alpha$ times the optimum cost.

\paragraph*{\PTWO\ (\ptwo).} 
  In this version, the adversary, in addition to giving requests, is obliged to help the algorithm by providing additional information as follows. The adversary must serve each request and reveal to the algorithm the subset of servers it moved. Note that the adversary does not reveal the destination to which it moved its servers -- revealing destinations makes the problem trivial because the algorithm can simply copy the adversary's movements. Given a request and the additional information about the adversary's server movements, the algorithm is required to move its own servers to cover the request. In an ideal scenario where the adversary serves the requests optimally, we require the algorithm to produce a solution whose cost competes with the cost of the optimal solution. However, consider a malicious adversary which, in an attempt to be as unhelpful to the algorithm as possible, produces a far-from-optimum solution and shares its information with the algorithm. In this case,  we do not require that the algorithm competes with the optimum solution -- such an algorithm would already solve the weighted $k$-server problem without the adversary's help. Instead, we require the algorithm to compete with the adversary's revealed solution. Formally, an algorithm is said to be $\alpha$-competitive if the (expected) cost of its output is at most $\alpha$ times the cost of the adversary's (possibly sub-optimal) solution. 



Obviously, an algorithm for the weighted $k$-server problem gives an algorithm for each of the above problems. Interestingly, we prove that the converse is also true. Formally,

\begin{restatable}[Composition Theorem]{theorem}{compthm}\label{two-problem-theorem}
If there exists an $\alpha_1$-competitive algorithm for \pone\ and there exists an $\alpha_2$-competitive algorithm for \ptwo, then there is an $(\alpha_1\alpha_2)$-competitive algorithm for the \WKS\ on uniform metrics.
\end{restatable}

We prove this theorem in~\Cref{sec_composition}. As a consequence of this theorem, it is enough to design $2^{\poly(k)}$-competitive algorithms for \pone\ and \ptwo\ to close the exponential gap between the upper and lower bounds on the randomized competitive ratio of weighted $k$-server. We already present such an algorithm for \ptwo\ in~\Cref{sec_wksha}. We prove, 

\begin{restatable}{theorem}{maintheorem}\label{Main-theorem}
There is a randomized algorithm for \ptwo\ with a competitive ratio of $2^{O(k^2)}$.
\end{restatable}

This reduces the task of designing a $2^{\poly(k)}$-competitive algorithm for weighted $k$-server to designing such an algorithm for \pone, a potentially easier problem.

Can we improve the $2^{O(k^2)}$ upper bound for \ptwo\ to, for example, $\poly(k)$? We answer this question in the negative. We show that, in fact, the lower bound construction by Ayyadevara and Chiplunkar~\cite{AyyadevaraC_ESA21} for weighted $k$-server applies to \ptwo\ too\footnote{The construction by Bansal et al.~\cite{BansalEK_FOCS17} applies too, and for the same reason, implying a doubly exponential lower bound on the deterministic competitive ratio of \ptwo.}, giving the following result. 

\begin{restatable}{theorem}{lowerbound}
The randomized competitive ratio of \ptwo\ is $\Omega(2^k)$.
\end{restatable}

The proof of this result is deferred to~\Cref{sec_lowerbound}.

\section{Preliminaries}\label{sec:prelims}


In this section we define the problems \pone\ and \ptwo\ formally, but before that, we restate the definitions of some terms introduced by Bansal et al.~\cite{BansalEK_FOCS17} which will be needed in our problem definitions.

\subsection{Service Patterns, Feasible Labelings, Extensions}

Throughout this paper, we assume without loss of generality that all the servers of the algorithm and the adversary move in response to the first request. Given a solution to an instance of the weighted $k$-server problem with $T$ requests, focus on the movements of the $\ell$'th server for an arbitrary $\ell$. The time instants at which these movements take place partition the interval $[1,T+1)$ into left-closed-right-open intervals so that the server stays put at some point during each of these intervals. Thus, ignoring the locations of the servers and focusing only on the time instants at which each server moves, we get a tuple of $k$ partitions of $[1,T+1)$, also known as a service pattern. Formally,

\begin{definition}[Service Pattern and Levels~\cite{BansalEK_FOCS17}]
A $k$-tuple $\II = (\II^1,\ldots,\II^k)$ is called a \textit{service pattern} over an interval $[t_{\text{begin}},t_{\text{end}})$ if each \aay{$\II^\ell$} is a partition of $[t_{\text{begin}},t_{\text{end}})$ comprising of left-closed-right-open intervals with integer boundaries. We call $\II^{\ell}$ the $\ell$'th \textit{level} of $\II$.
\end{definition}

Observe that the cost of a solution is completely determined by its service pattern $\II = (\II^1,\ldots,\II^k)$: the cost equals the sum over $\ell\in\{1,\ldots,k\}$ of the number of intervals in $\II^{\ell}$ times the weight of the $\ell$'th server.

In order to completely specify a solution, in addition to a service pattern $\II = (\II^1,\ldots,\II^k)$, we need to specify for each $\ell$ and each interval $I\in\II^{\ell}$ the location of the $\ell$'th server during the time interval $I$. We refer to this assignment as a labeling of the service pattern. Moreover, to serve the $t$'th request $\sigma_t$, we need at least one server to occupy $\sigma_t$ at time $t$, that is, we need that there exists a level of $\II$ in which the (unique) interval containing $t$ is labeled $\sigma_t$. Such a labeling is called a feasible labeling. Formally,

\begin{definition}[Labeling and Feasibility~\cite{BansalEK_FOCS17}]\label{feasible-labeling}
A \textit{labeling} of a service pattern $\II= (\II^1,\ldots,\II^k)$ is a function from the multi-set $\II^1\uplus\cdots\uplus\II^k$ to the set $U$ of points in the metric space.
We say that a labeling $\gamma$ of $\II$ is \textit{feasible} with respect to a request sequence $\rho=(\sigma_1,\ldots,\sigma_T)$, if for each time $t$, there exists an interval $I\in\II^1\uplus\cdots\uplus\II^k$ containing $t$ such that $\gamma(I) = \sigma_t$. We say that a service pattern $\II$ is \textit{feasible} with respect to $\rho$ if there exists a feasible labeling of $\II$ with respect to $\rho$. 
\end{definition}

Recall that in the definition of the weighted $k$-server problem in \Cref{sec:intro}, we assumed that the servers are numbered in a non-decreasing order of their weights. Consider the more interesting case where the weights increase at least geometrically. If we enforce that every time a server moves, all servers lighter than it move too, we lose at most a constant factor in the competitive ratio. The advantage of this enforcement is that we have a more structured class of service patterns, called hierarchical service patterns.

\begin{definition}[Hierarchical Service Pattern~\cite{BansalEK_FOCS17}]
A service pattern $\II=(\II^1,\ldots,\II^k)$ is \textit{hierarchical} if for every $\ell\in\{1,\ldots,k-1\}$, the partition $\II^\ell$ refines the partition $\II^{\ell+1}$.
\end{definition}

Clearly, any service pattern can be made hierarchical in an online manner with at most a $k$ factor loss in the cost. Since we aim to obtain $2^{\poly(k)}$-competitive algorithms, the $k$ factor loss is affordable, and therefore, we only consider hierarchical service patterns throughout this paper. Henceforth, by service pattern we actually mean a hierarchical service pattern.

Next, consider some online algorithm for the weighted $k$-server problem and the solution it outputs on some request sequence. For each $t$, let $\II_t$ denote the service pattern corresponding to the algorithm's solution until the $t$'th request. Observe that $\II_{t-1}$ and $\II_t$ are closely related: if the algorithm moves the lightest $\ell$ servers to serve the $t$'th request (possibly $\ell=0$), then the interval $[t,t+1)$ gets added to the first $\ell$ levels of $\II_{t-1}$, whereas in each of the remaining levels, the last interval in the level merges with $[t,t+1)$. We call $\II_t$ the $\ell$-extension of $\II_{t-1}$. More formally,

\begin{definition}[$\ell$-extension]
Let $\II_{t-1} = (\II_{t-1}^1,\ldots,\II_{t-1}^k)$ be a hierarchical service pattern over the interval $[1,t)$, and let $L_{t-1}^i$ 
be the last interval in $\II^i_{t-1}$. For $\ell \in \{0,\ldots,k\}$, we define the $\ell$-extension of $\II_{t-1}$ to be the service pattern $\II_t = (\II_t^1,\ldots,\II_t^k)$ over the interval $[1,t+1)$ where:
\begin{itemize}
    \item $\forall i \leq \ell,\  \II_t^i = \II_{t-1}^i \cup \{[t,t+1)\}$.
    \item $\forall i > \ell,\ \II_t^i = (\II_{t-1}^i \setminus L_{t-1}^i) \cup \{L_{t-1}^i\cup[t,t+1)\}$.
\end{itemize}
\end{definition}


Observe that the $\ell$-extension of a hierarchical service pattern is a hierarchical service pattern.

\subsection{Problem Definitions}

Recall that our core idea to solve \WKS\ problem is to construct an algorithm using algorithms for its two relaxed versions. We defined them informally in \Cref{sec:intro}. Their formal definitions are as follows.

\begin{definition}[\PONE\ (\pone)]\label{def:p1}
    For every online request $\sigma_t\in U$, an algorithm for \pone\ is required to output a service pattern $\II_{t}$, which is the $\ell_t$-extension of $\II_{t-1}$ for some $\ell_t \in \{0,\ldots,k\}$, such that $\II_{t}$ is feasible with respect to the request sequence $\sigma_1,\sigma_2,\ldots,\sigma_t$. Equivalently, the algorithm outputs $\ell_t$ for each $t$.
    An algorithm for \pone\ is said to be $\alpha$-competitive if the (expected) cost of the algorithm's service pattern is at most $\alpha$ times the optimal cost.
\end{definition}

\begin{definition}[\PTWO\ (\ptwo)]\label{def:p2}
    For every online request $\sigma_t\in U$, the adversary reveals a service pattern $\II_{t}$, which is the $\ell_t$-extension of $\II_{t-1}$ for some $\ell_t \in \{0,\ldots,k\}$, such that $\II_{t}$ is feasible with respect to the request sequence $\sigma_1,\sigma_2,\ldots,\sigma_t$. Equivalently, the algorithm's input is the pair $(\sigma_t, \ell_t)$.
    An algorithm for \ptwo\ is required to serve the request $\sigma_t$, i.e., move servers to ensure that $\sigma_t$ is covered by some server. An algorithm for \ptwo\ is said to be $\beta$-competitive if the (expected) cost of the algorithm's solution is at most $\beta$ times the cost of the final service pattern revealed by the adversary.
\end{definition}

\section{The Composition Theorem}\label{sec_composition}

In this section, we explain how we can construct a weighted $k$-server algorithm using algorithms for its two relaxations -- \pone\ and \ptwo\footnote{\aay{On a high level, our construction resembles the result by Ben-David et al.~\cite{Ben-DavidBKTW94}, which states that if there is an $\alpha_1$-competitive randomized algorithm against online adversary and an $\alpha_2$-competitive algorithm against any oblivious adversary, then there is an $(\alpha_1\alpha_2)$ competitive randomized algorithm for any adaptive offline adversary.}}.

\compthm*

\begin{proof}
Let $\AA_1$ be an $\alpha_1$-competitive algorithm for \pone, and $\AA_2$ be an $\alpha_2$-competitive algorithm for \ptwo.
Our algorithm $\AA$ for \WKS\ internally runs the two algorithms $\AA_1$ and $\AA_2$. At all times, $\AA$ keeps each of its servers at the same point where the corresponding server of $\AA_2$ is located. For every input request $\sigma_t$, $\AA$ performs the following sequence of steps.
\begin{enumerate}
    \item $\AA$ passes $\sigma_t$ to $\AA_1$.
    \item In response, $\AA_1$ outputs an $\ell_t$ such that the service pattern $\II_{t}$, which is the $\ell_t$-extension to $\II_{t-1}$, is feasible for the request sequence $\sigma_1,\sigma_2, \ldots,\sigma_t$.
    \item $\AA$ passes $(\sigma_t,\ell_t)$ to $\AA_2$.
    \item In response, $\AA_2$ moves its servers to serve the request $\sigma_t$.
    \item $\AA$ copies the movements of $\AA_2$'s servers.
\end{enumerate}

To analyze the competitiveness of $\AA$, consider an arbitrary sequence $\rho$ of requests, and let $T$ denote its length. Let $\text{OPT}$ denote the cost of an optimal solution for $\rho$. Denote the cost of a service pattern $\II$ by $\text{cost}(\II)$. Recall that $\II_T$ is the final service pattern output by $\AA_1$. Let $\II'_T$ denote the service pattern corresponding to $\AA_2$'s output. Note that $\II_T$ and $\II'_T$ are random variables, and since $\AA$'s output is same as $\AA_2$'s output, the cost of $\AA$'s output is $\text{cost}(\II'_T)$.

Since the output of $\AA_1$ is a sequence of extensions such that the service pattern remains feasible with respect to the request sequence at all times, the sequence $(\sigma_t,\ell_t)_{t=1,\ldots,T}$ is a valid instance of \ptwo\ (with probability one over the randomness of $\AA_1$). Since $\AA_2$ is $\alpha_2$-competitive, we have $\mathbb{E}[\text{cost}(\II'_T)\mid\II_T=\II]\leq\alpha_2\cdot\text{cost}(\II)$ for every service pattern $\II$ feasible with respect to $\rho$. This implies $\mathbb{E}[\text{cost}(\II'_T)]\leq\alpha_2\cdot\mathbb{E}[\text{cost}(\II_T)]$. Since $\AA_1$ is $\alpha_1$-competitive, $\mathbb{E}[\text{cost}(\II_T)]\leq\alpha_1\cdot\text{OPT}$. Thus, $\mathbb{E}[\text{cost}(\II'_T)]\leq\alpha_1\alpha_2\cdot\text{OPT}$. This implies that $\AA$ is $(\alpha_1\alpha_2)$-competitive.
\end{proof}

\section{Competing with a Revealed Service Pattern}\label{sec_wksha}

We organize this section as follows. In~\Cref{sec:struct}, we prove some structural results that are used in the definition and analysis of our algorithm. We define our algorithm formally in~\Cref{sec:algwksha} and analyze its competitive ratio in~\Cref{sec:comptv}. We use the following notation.
\begin{itemize}
\item $U$ denotes the set of points in a uniform metric space.
    \item For $t$ from $1$ to $T$, The $t$'th request is $\sigma_t\in U$, and $\rho_t=(\sigma_1,\ldots,\sigma_t)$ denotes the sequence of requests received until time $t$.
    \item The service pattern revealed by the adversary with the $t$'th request is denoted by $\II_t=(\II_t^1,\ldots,\II_t^k)$. Recall that $\II_t$ is the  $\ell_t$-extension of $\II_{t-1}$. Without loss of generality, we assume that the adversary moves all its servers with the first request, and therefore $\ell_1=k$.
    \item $L_t^{\ell}$ denotes the last interval in $\II_t^{\ell}$, that is, the unique interval in $\II_t^{\ell}$ that covers $[t,t+1)$.
    \item $s_t^{\ell}$ denotes the location of our algorithm's $\ell$'th server after processing the $t$'th request. Since our algorithm is randomized, $s_t^{\ell}$ is a random variable. Note that for the $t$'th request to be served, we must have $\sigma_t\in\{s_t^1,\ldots,s_t^k\}$ with probability one.
\end{itemize}

Consider the adversary's service pattern $\II_t=(\II_t^1,\ldots,\II_t^k)$. For an arbitrary $\ell$, fix the labels of the last intervals $L_t^{\ell+1},\ldots,L_t^k$ in the top $k-\ell$ levels $\II_t^{\ell+1},\ldots,\II_t^k$ of $\II_t$, and consider all feasible labelings of $\II_t$ with respect to $\rho$ that agree with the fixed labels. The set of labels that these labelings assign to the last interval $L^{\ell}_t$ of $\II^{\ell}_t$ will be crucial for our algorithm. We now define this set formally.

\begin{definition}\label{feasible-labels-set}
For any $t\in \{1,\ldots, T\}$, $\ell\in\{1,\ldots,k\}$, and $p^{\ell+1},\ldots,p^k\in U$, the set
$Q_t^{\ell}(p^{\ell+1},\ldots,p^k)$ is defined to be the set of points $p^{\ell}$ for which there exists a feasible labeling $\gamma$ of $\II_t$ with respect to $\rho_t$ such that $\gamma(L^i_t) = p^i$ for all $i\in\{\ell,\ldots,k\}$.
\end{definition}

\subsection{Structural Results}\label{sec:struct}

Bansal et al.~\cite{BansalEK_FOCS17} considered the following combinatorial question: given a service pattern $\II$ and a request sequence $\rho$, how many labels can an interval in the $k$'th level of $\II$ get, over all possible feasible labelings of $\II$ with respect to $\rho$? They derived the following interesting property.

\begin{fact}[Dichotomy Property~\cite{BansalEK_FOCS17}]\label{dichotomy-theorem}
There exists a sequence $n_1,n_2,\ldots$ of integers with $n_k\leq2^{2^{k+3\log k}}$ such that the following holds: for every $k$, every sequence of requests $\rho=(\sigma_1,\ldots,\sigma_T)$, every service pattern $\II=(\II^1,\ldots,\II^k)$ over $[1,T+1)$, and every $I\in \II^k$, the set $Q$ of labels of $I$ over all feasible labelings of $\II$ with respect to $\rho$ is the entire $U$, or it has size at most $n_k$.
\end{fact}

The next lemma generalizes the above result to intervals in every level.
\begin{lemma}[Generalized Dichotomy Property]\label{dichotomy-theorem-extn}
For every $t\in \{1,\ldots,T\}$, $\ell\in \{1,\ldots, k\}$, and $p^{\ell+1},\ldots,p^k\in U$, the set $Q^{\ell}_t(p^{\ell+1},\ldots,p^k)$ is the entire $U$, or it has size at most $n_{\ell}$, where $n_{\ell}\leq2^{2^{\ell+3\log \ell}}$ is the constant from \Cref{dichotomy-theorem}.
\end{lemma}

\begin{proof}
The lemma holds trivially when $Q^{\ell}_t(p^{\ell+1},\ldots,p^k) =\emptyset$, so assume $Q^{\ell}_t(p^{\ell+1},\ldots,p^k)\neq\emptyset$. Let $\rho'$ denote the request sequence during the interval $L^{\ell}_t$ and $\JJ$ denote the restriction of the service pattern $\II_t$ to the interval $L^{\ell}_t$ and levels $1,\ldots,\ell$. Let $\rho''$ denote the subsequence of $\rho'$ formed by removing all the requests to $p^{\ell+1},\ldots,p^k$. Let $Q$ denote the set of labels of the interval $L^{\ell}_t$ over all the feasible labelings of the $\ell$-level service pattern $\JJ$ with respect to $\rho''$. From \Cref{dichotomy-theorem} we get that the set $Q$ is either $U$ or has size at most $n_{\ell}$. We argue that $Q^{\ell}_t(p^{\ell+1},\ldots,p^k) = Q$, and this implies the claim. 
Refer to \Cref{fig:extn} for a working illustration on an instance with $k=5$ and $\ell=3$.

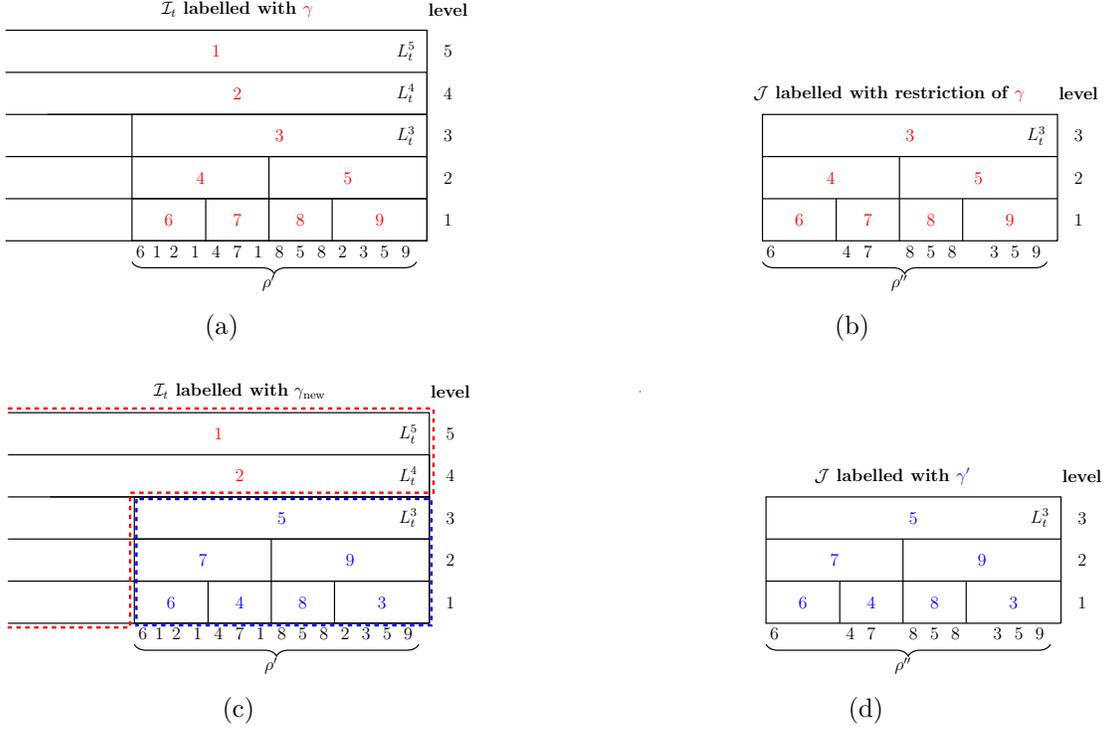
\begin{figure}[h]
	\captionsetup[subfigure]{justification=centering}
	\hspace{-10pt}
	\begin{subfigure}[b]{0.5\textwidth}
		\centering
		\resizebox{!}{4cm}{
		\begin{tikzpicture}	
            \draw(-.7,-1);
	    	\draw[thick](0,0)--(10,0)--(10,5)--(0,5);
            \draw[thick](0,4)--(10,4)--(10,3)--(1,3);
	    	
	    	\draw (5.5, 5.5) node {\large \textbf{$\II_t$ labelled with \textcolor{red}{$\mathbf{\gamma}$}}};
	    	\draw[thick](0,3)--(10,3);
	    	\draw[thick](0,2)--(10,2);
	    	\draw[thick](0,1)--(10,1);
	    	
	    	\draw (10.5, 5.5) node {\large \textbf{level}};
			\draw (10.5, 4.5) node {\large \textbf{$5$}};
			\draw (10.5, 3.5) node {\large \textbf{$4$}};
			\draw (10.5, 2.5) node {\large \textbf{$3$}};
			\draw (10.5, 1.5) node {\large \textbf{$2$}};
			\draw (10.5, .5) node {\large \textbf{$1$}};

			\draw[thick](10,0)--(10,3)--(3,3)--(3,0)--(10,0);
			\draw[thick](10,2)--(3,2);
			\draw[thick](10,1)--(3,1);
			
			\draw[thick](4.75,0)--(4.75,1);
			\draw[thick](7.75,0)--(7.75,1);
			
			\draw[thick](6.25,0)--(6.25,2);
			
			\draw[thick,decorate,decoration={brace,amplitude=8pt}] (9.75,-.5) -- (3,-.5);
			\draw (6.25,-1) node {\large \textbf{$\rho'$}}; 
			
			\draw (9.5,4.5) node {\large \textbf{$L_t^5$}}; 
			\draw (9.5,3.5) node {\large \textbf{$L_t^4$}}; 
			\draw (9.5,2.5) node {\large \textbf{$L_t^3$}}; 
			
			\draw[red] (5,4.5) node {\large \textbf{$1$}}; 
			\draw[red] (5.5,3.5) node {\large \textbf{$2$}}; 
			\draw[red] (6.5,2.5) node {\large \textbf{$3$}}; 
			
			\draw[red] (4.625,1.5) node {\large \textbf{$4$}}; 
			\draw[red] (8.125,1.5) node {\large \textbf{$5$}}; 
			
			\draw[red] (3.875,.5) node {\large \textbf{$6$}}; 
			\draw[red] (5.5,.5) node {\large \textbf{$7$}}; 
			\draw[red] (7,.5) node {\large \textbf{$8$}}; 
			\draw[red] (8.875,.5) node {\large \textbf{$9$}}; 
			
			\draw (3.2,-.25) node {\large \textbf{$6$}};
			\draw (3.6,-.25) node {\large \textbf{$1$}};
			\draw (4,-.25) node {\large \textbf{$2$}};
			\draw (4.5,-.25) node {\large \textbf{$1$}};
			
			\draw (5,-.25) node {\large \textbf{$4$}};
			\draw (5.5,-.25) node {\large \textbf{$7$}};
			\draw (6,-.25) node {\large \textbf{$1$}};
			
			\draw (6.5,-.25) node {\large \textbf{$8$}};
			\draw (7,-.25) node {\large \textbf{$5$}};
			\draw (7.5,-.25) node {\large \textbf{$8$}};
			
			\draw (8,-.25) node {\large \textbf{$2$}};
			\draw (8.5,-.25) node {\large \textbf{$3$}};
			\draw (9,-.25) node {\large \textbf{$5$}};
			\draw (9.5,-.25) node {\large \textbf{$9$}};
			
		\end{tikzpicture}}
		\caption{}
		\label{fig:extn1}
	\end{subfigure}
	\begin{subfigure}[b]{0.5\textwidth}
		\centering
		\resizebox{!}{4cm}{
		\begin{tikzpicture}		
            \draw(-.7,-1);
	        \draw[white] (1, 5.5) node {\large $\mathbf{k=5,\ell=3}$};
	    	\draw (6, 3.5) node {\large \textbf{$\JJ$ labelled with restriction of \textcolor{red}{$\mathbf{\gamma}$}}};
	    	
	    	\draw (10.5, 3.5) node {\large \textbf{level}};
			\draw (10.5, 2.5) node {\large \textbf{$3$}};
			\draw (10.5, 1.5) node {\large \textbf{$2$}};
			\draw (10.5, .5) node {\large \textbf{$1$}};

			\draw[thick](10,0)--(10,3)--(3,3)--(3,0)--(10,0);
			\draw[thick](10,2)--(3,2);
			\draw[thick](10,1)--(3,1);
			
			\draw[thick](4.75,0)--(4.75,1);
			\draw[thick](7.75,0)--(7.75,1);
			
			\draw[thick](6.25,0)--(6.25,2);
			
			\draw[thick,decorate,decoration={brace,amplitude=8pt}] (9.75,-.5) -- (3,-.5);
			\draw (6.25,-1) node {\large \textbf{$\rho''$}}; 
			
			\draw (9.5,2.5) node {\large \textbf{$L_t^3$}}; 
			
			\draw[red] (6.5,2.5) node {\large \textbf{$3$}}; 
			
			\draw[red] (4.625,1.5) node {\large \textbf{$4$}}; 
			\draw[red] (8.125,1.5) node {\large \textbf{$5$}}; 
			
			\draw[red] (3.875,.5) node {\large \textbf{$6$}}; 
			\draw[red] (5.5,.5) node {\large \textbf{$7$}}; 
			\draw[red] (7,.5) node {\large \textbf{$8$}}; 
			\draw[red] (8.875,.5) node {\large \textbf{$9$}}; 
			
			\draw (3.2,-.25) node {\large \textbf{$6$}};
			
			\draw (5,-.25) node {\large \textbf{$4$}};
			\draw (5.5,-.25) node {\large \textbf{$7$}};
			
			\draw (6.5,-.25) node {\large \textbf{$8$}};
			\draw (7,-.25) node {\large \textbf{$5$}};
			\draw (7.5,-.25) node {\large \textbf{$8$}};
	
			\draw (8.5,-.25) node {\large \textbf{$3$}};
			\draw (9,-.25) node {\large \textbf{$5$}};
			\draw (9.5,-.25) node {\large \textbf{$9$}};
			
		\end{tikzpicture}}
		\caption{}
		\label{fig:extn2}
	\end{subfigure}
    \par\bigskip
	\begin{subfigure}[b]{0.5\textwidth}
		\centering
		\resizebox{!}{4cm}{
		\begin{tikzpicture}	
            
	        \draw[thick](0,0)--(10,0)--(10,5)--(0,5);
            \draw[thick](0,4)--(10,4)--(10,3)--(1,3);
	    	
	    	\draw (5.5, 5.5) node {\large \textbf{$\II_t$ labelled with} $\gamma_{\text{new}}$};
	    	\draw[thick](0,3)--(10,3);
	    	\draw[thick](0,2)--(10,2);
	    	\draw[thick](0,1)--(10,1);
	    	
	    	\draw (10.5, 5.5) node {\large \textbf{level}};
			\draw (10.5, 4.5) node {\large \textbf{$5$}};
			\draw (10.5, 3.5) node {\large \textbf{$4$}};
			\draw (10.5, 2.5) node {\large \textbf{$3$}};
			\draw (10.5, 1.5) node {\large \textbf{$2$}};
			\draw (10.5, .5) node {\large \textbf{$1$}};

			\draw[thick](10,0)--(10,3)--(3,3)--(3,0)--(10,0);
			\draw[thick](10,2)--(3,2);
			\draw[thick](10,1)--(3,1);
			
			\draw[thick](4.75,0)--(4.75,1);
			\draw[thick](7.75,0)--(7.75,1);
			
			\draw[thick](6.25,0)--(6.25,2);
			
			\draw[thick,decorate,decoration={brace,amplitude=8pt}] (9.75,-.5) -- (3,-.5);
			\draw (6.25,-1) node {\large \textbf{$\rho'$}}; 
			
			\draw (9.5,4.5) node {\large \textbf{$L_t^5$}}; 
			\draw (9.5,3.5) node {\large \textbf{$L_t^4$}}; 
			\draw (9.5,2.5) node {\large \textbf{$L_t^3$}}; 
			
			\draw[red] (5,4.5) node {\large \textbf{$1$}}; 
			\draw[red] (5.5,3.5) node {\large \textbf{$2$}}; 
			\draw[blue] (6.5,2.5) node {\large \textbf{$5$}}; 
			
			\draw[blue] (4.625,1.5) node {\large \textbf{$7$}}; 
			\draw[blue] (8.125,1.5) node {\large \textbf{$9$}}; 
			
			\draw[blue] (3.875,.5) node {\large \textbf{$6$}}; 
			\draw[blue] (5.5,.5) node {\large \textbf{$4$}}; 
			\draw[blue] (7,.5) node {\large \textbf{$8$}}; 
			\draw[blue] (8.875,.5) node {\large \textbf{$3$}}; 
			
			\draw (3.2,-.25) node {\large \textbf{$6$}};
			\draw (3.6,-.25) node {\large \textbf{$1$}};
			\draw (4,-.25) node {\large \textbf{$2$}};
			\draw (4.5,-.25) node {\large \textbf{$1$}};
			
			\draw (5,-.25) node {\large \textbf{$4$}};
			\draw (5.5,-.25) node {\large \textbf{$7$}};
			\draw (6,-.25) node {\large \textbf{$1$}};
			
			\draw (6.5,-.25) node {\large \textbf{$8$}};
			\draw (7,-.25) node {\large \textbf{$5$}};
			\draw (7.5,-.25) node {\large \textbf{$8$}};
			
			\draw (8,-.25) node {\large \textbf{$2$}};
			\draw (8.5,-.25) node {\large \textbf{$3$}};
			\draw (9,-.25) node {\large \textbf{$5$}};
			\draw (9.5,-.25) node {\large \textbf{$9$}};
			
		    \draw[ultra thick, dashed, red](0,5.1)--(10.1,5.1)--(10.1,3.1)--(2.9,3.1)--(2.9,-.1)--(0,-.1);
		    \draw[ultra thick, dashed, blue](3.05,2.95)--(3.05,-.05)--(10.05,-.05)--(10.05,2.95)--(3.05,2.95);

		\end{tikzpicture}}
		\caption{}
		\label{fig:extn3}
	\end{subfigure}
	\begin{subfigure}[b]{0.5\textwidth}
		\centering
		\resizebox{!}{4cm}{
		\begin{tikzpicture}			
	        \draw[white] (1, 5.5) node {\large $\mathbf{k=5,\ell=3}$};
	        \draw(0,5.5)--(0,5.5);
	    	\draw (6, 3.5) node {\large \textbf{$\JJ$ labelled with \textcolor{blue}{$\mathbf{\gamma'}$}}};
	    	
	    	\draw (10.5, 3.5) node {\large \textbf{level}};
			\draw (10.5, 2.5) node {\large \textbf{$3$}};
			\draw (10.5, 1.5) node {\large \textbf{$2$}};
			\draw (10.5, .5) node {\large \textbf{$1$}};

			\draw[thick](10,0)--(10,3)--(3,3)--(3,0)--(10,0);
			\draw[thick](10,2)--(3,2);
			\draw[thick](10,1)--(3,1);
			
			\draw[thick](4.75,0)--(4.75,1);
			\draw[thick](7.75,0)--(7.75,1);
			
			\draw[thick](6.25,0)--(6.25,2);
			
			\draw[thick,decorate,decoration={brace,amplitude=8pt}] (9.75,-.5) -- (3,-.5);
			\draw (6.25,-1) node {\large \textbf{$\rho''$}}; 
			
			\draw (9.5,2.5) node {\large \textbf{$L_t^3$}}; 
			
			\draw[blue] (6.5,2.5) node {\large \textbf{$5$}}; 
			
			\draw[blue] (4.625,1.5) node {\large \textbf{$7$}}; 
			\draw[blue] (8.125,1.5) node {\large \textbf{$9$}}; 
			
			\draw[blue] (3.875,.5) node {\large \textbf{$6$}}; 
			\draw[blue] (5.5,.5) node {\large \textbf{$4$}}; 
			\draw[blue] (7,.5) node {\large \textbf{$8$}}; 
			\draw[blue] (8.875,.5) node {\large \textbf{$3$}}; 
			
			\draw (3.2,-.25) node {\large \textbf{$6$}};
			
			\draw (5,-.25) node {\large \textbf{$4$}};
			\draw (5.5,-.25) node {\large \textbf{$7$}};
			
			\draw (6.5,-.25) node {\large \textbf{$8$}};
			\draw (7,-.25) node {\large \textbf{$5$}};
			\draw (7.5,-.25) node {\large \textbf{$8$}};
	
			\draw (8.5,-.25) node {\large \textbf{$3$}};
			\draw (9,-.25) node {\large \textbf{$5$}};
			\draw (9.5,-.25) node {\large \textbf{$9$}};
			
		\end{tikzpicture}}
		\caption{}
		\label{fig:extn4}
	\end{subfigure}
	\caption{An illustration of \Cref{dichotomy-theorem-extn} for $k=5$ and $\ell=3$. (a) Depicts $\II_t$ with a labeling  $\gamma$ (colored in red) feasible with respect to $\rho_t$ such that $\gamma(L^5_t)=1$ and $\gamma(L^4_t)=2$. (b) Depicts the $3$-level service pattern $\JJ$ labeled with the restriction of $\gamma$, along with $\rho''$, the subsequence of $\rho'$ formed by removing all the requests to points $1$ and $2$. (d) Depicts the labeling $\gamma'$ of $\JJ$ feasible with respect to $\rho''$. (c) Shows the new labeling $\gamma_{\text{new}}$ constructed by overwriting $\gamma'$ onto $\gamma$ for every interval in $\JJ$.}
	\label{fig:extn}
\end{figure}

For any $p^{\ell}\in Q^{\ell}_t(p^{\ell+1},\ldots,p^k)$, from \Cref{feasible-labels-set} we get that there exists a feasible labeling $\gamma$ of $\II_t$ with respect to $\rho_t$ such that $\gamma(L^i_t)=p^i$ for all $i\in\{\ell,\ldots,k\}$. In the labeling $\gamma$, the servers $\ell+1,\ldots,k$ can only serve the requests to the points $p^{\ell+1},\ldots,p^k$ during the interval $L_t^{\ell}$. This implies that all the requests in $\rho''$ must be served by servers $1,\ldots,\ell$. Thus, the restriction $\gamma'$ of $\gamma$ to $\JJ$ is feasible with respect to $\rho''$. But $\gamma'(L_t^{\ell})=\gamma(L_t^{\ell})=p^{\ell}$. Therefore, $p^{\ell}\in Q$. Thus, $Q^{\ell}_t(p^{\ell+1},\ldots,p^k)\subseteq Q$.


Now consider any point $p^{\ell}\in Q$, and let $\gamma'$ be a feasible labeling of $\JJ$ with respect to $\rho''$ such that $\gamma'(L_t^{\ell})=p^{\ell}$. Suppose $\gamma$ is a feasible labeling of $\II_t$ with respect to $\rho_t$ such that $\gamma(L^i_t) = p^i$ for all $i\in\{\ell+1,\ldots,k\}$ (such a labeling exists because we assumed $Q^{\ell}_t(p^{\ell+1},\ldots,p^k)\neq\emptyset$). Overwrite the labeling $\gamma'$ onto $\gamma$ to get a new labeling $\gamma_{\text{new}}$. Formally, $\gamma_{\text{new}}(I) = \gamma'(I)$ if interval $I$ is in $\JJ$, else $\gamma_{\text{new}}(I)=\gamma(I)$. We claim that $\gamma_{\text{new}}$ is also a feasible labeling of $\II_t$ with respect to $\rho_t$. This can be argued as follows.

The labeling $\gamma_{\text{new}}$ serves all the requests in $\rho''$ because it agrees with $\gamma'$ in the service pattern $\JJ$. $\gamma_{\text{new}}$ serves all the requests during the interval $L_t^{\ell}$ other than those in $\rho''$ because all these requests are made at points in $\{p^{\ell+1},\ldots,p_k\}$, and $\gamma_{\text{new}}(L_t^i)=\gamma(L_t^i)=p^i$ for all $i\in\{\ell+1,\ldots,k\}$. Finally, $\gamma_{\text{new}}$ serves all the requests before the interval $L_t^{\ell}$ because it agrees with $\gamma$ before the interval $L_t^{\ell}$.

Thus, $\gamma_{\text{new}}$ is a feasible labeling of $\II_t$ with respect to $\rho_t$ such that $\gamma_{\text{new}}(L^i_t) = p^i$ for all $i\in\{\ell+1,\ldots,k\}$. From \Cref{feasible-labels-set}, we get that $\gamma_{\text{new}}(L_t^{\ell})=p^{\ell}\in Q^{\ell}_t(p^{\ell+1},\ldots,p^k)$. This implies $Q\subseteq Q^{\ell}_t(p^{\ell+1},\ldots,p^k)$.
\end{proof}

We now state a useful consequence of the simple fact that the restriction of a solution for the first $t$ requests to the first $t-1$ requests is a solution for the first $t-1$ requests.
\begin{lemma}\label{subset-feasible-labels}
For every $t\in \{2,\ldots,T\}$, $\ell\in \{\ell_t+1,\ldots, k\}$, and $p^{\ell+1},\ldots,p^k\in U$, $Q^{\ell}_t(p^{\ell+1},\ldots,p^k)\subseteq Q^{\ell}_{t-1}(p^{\ell+1},\ldots,p^k)$.
\end{lemma}

\begin{proof}
Recall that $\II_t=(\II_t^1,\ldots,\II_t^k)$ is the $\ell_t$-extension of $\II_{t-1}=(\II_{t-1}^1,\ldots,\II_{t-1}^k)$, which means $L_{t-1}^{\ell}=L_t^{\ell}\setminus[t,t+1)\neq\emptyset$. 
By \Cref{feasible-labels-set}, if some point $p^{\ell}$ is in $Q^{\ell}_t(p^{\ell+1},\ldots,p^k)$, then there exists a feasible labeling $\gamma$ of $\II_t$ with respect to $\rho_t$ such that $\gamma(L_t^i)=p^i$ for all $i\in\{\ell,\ldots,t\}$. Restrict $\gamma$ to obtain a labeling $\gamma'$ of $\II_{t-1}$ in the obvious manner: $\gamma'(L_{t-1}^i)=\gamma(L_t^i)=p^i$ for all $i\in\{\ell,\ldots,t\}$ and $\gamma'(I)=\gamma(I)$ for all other intervals $I$ of $\II_{t-1}$ (which are also intervals of $\II_t$). It is easy to check that $\gamma'$ is a feasible labeling of $\II_{t-1}$ with respect to $\rho_{t-1}$. Thus, from \Cref{feasible-labels-set} we get, $p^{\ell}\in Q^{\ell}_{t-1}(p^{\ell+1},\ldots,p^k)$.
\end{proof}

\subsection{Algorithm} \label{sec:algwksha}

Before stating our \ptwo\ algorithm formally, we give some intuitive explanation. Recall that $s_t^{\ell}$ denotes the location of the algorithm's $\ell$'th server after serving the $t$'th request. The critical invariant maintained by our algorithm is the following.
\begin{invariant}\label{inv:random-point}
    For every $t$ and $\ell$, in response to the $t$'th request, the algorithm keeps its $\ell$'th server at a uniformly random point in $Q_t^{\ell}(s_t^{\ell+1},\ldots,s_t^k)$.
\end{invariant} 
\Cref{random-point} states this claim formally. For now, let us just understand the scenarios in which the algorithm needs to move the $\ell$'th server at time $t$ so that it occupies \emph{some} point in $Q_t^{\ell}(s_t^{\ell+1},\ldots,s_t^k)$. These scenarios are as follows.
\begin{enumerate}
    \item The algorithm moves the $\ell'$'th server for some $\ell'>\ell$. This means that, potentially, $s_t^{\ell'}\neq s_{t-1}^{\ell'}$, so $Q_t^{\ell}(s_t^{\ell+1},\ldots,s_t^k)$ could be different from $Q_{t-1}^{\ell}(s_{t-1}^{\ell+1},\ldots,s_{t-1}^k)$.
    \item $\ell_t\geq\ell$. This means that $L_t^{\ell}=[t,t+1)$ is disjoint from $L_{t-1}^{\ell}$, and again, potentially, $Q_t^{\ell}(s_t^{\ell+1},\ldots,s_t^k)$ could be different from $Q_{t-1}^{\ell}(s_{t-1}^{\ell+1},\ldots,s_{t-1}^k)$.
    \item None of the above happens, so by \Cref{subset-feasible-labels}, $Q_t^{\ell}(s_t^{\ell+1},\ldots,s_t^k)\subseteq Q_{t-1}^{\ell}(s_{t-1}^{\ell+1},\ldots,s_{t-1}^k)$. However, $s_{t-1}^{\ell}\notin Q_t^{\ell}(s_t^{\ell+1},\ldots,s_t^k)$ (that is, $s_{t-1}^{\ell}\in Q_{t-1}^{\ell}(s_{t-1}^{\ell+1},\ldots,s_{t-1}^k)\setminus Q_t^{\ell}(s_t^{\ell+1},\ldots,s_t^k)$), so the $\ell$'th server can no longer remain at the same place $s_{t-1}^{\ell}$ as before.
\end{enumerate}
We call the movement in the first two scenarios a \textit{forced movement} (because the movement of some other server forced this movement), and the movement in the third scenario an \textit{unforced movement}. If none of the above scenarios arises, then the $\ell$'th server stays put. \Cref{alg:p2} is the formal description of our algorithm for \ptwo.

\begin{algorithm}
\caption{\ptwo}\label{alg:p2}
    \begin{algorithmic}[1]
        \FOR{$t=1$ to $T$}
            \STATE \textbf{Input:} request $\sigma_t\in U$, and $\ell_t\in\{0,\ldots,k\}$.
            \STATE \COMMENT{Recall: $\II_t$ is the $\ell_t$-extension of $\II_{t-1}$.}
            \STATE \texttt{flag} $\gets$ \texttt{FALSE}
            \FOR{$\ell=k$ to $1$}
                \STATE \COMMENT{Decide movements of servers in decreasing order of weight.}
                \STATE \COMMENT{\texttt{flag} $=$ \texttt{TRUE} indicates that an unforced movement of some server heavier than the $\ell$'th has happened.}
                \STATE Compute $Q^{\ell}_t(s_t^{\ell+1},\ldots,s_t^k)$ (by brute force).
                \IF{\texttt{flag} \texttt{OR} $\ell\leq\ell_t$} \label{alg:if1}
                    \STATE $s_t^{\ell}$ $\gets$ a uniformly random point in $Q^{\ell}_t(s_t^{\ell+1},\ldots,s_t^k)$. \COMMENT{forced movement} \label{alg:mov1}
                \ELSIF{$s^{\ell}_{t-1}\notin Q^{\ell}_t(s_t^{\ell+1},\ldots,s_t^k)$} \label{alg:if2}
                    \STATE $s_t^{\ell}$ $\gets$ a uniformly random point in $Q^{\ell}_t(s_t^{\ell+1},\ldots,s_t^k)$. \COMMENT{unforced movement} \label{alg:mov2}
                    \STATE \texttt{flag} $\gets$ \texttt{TRUE}.
                \ELSE
                    \STATE $s^{\ell}_t$ $\gets$ $s^{\ell}_{t-1}$. \COMMENT{no movement} \label{alg:no-movement}
                \ENDIF
            \ENDFOR
        \ENDFOR
    \end{algorithmic}
\end{algorithm}

Note that it is unclear so far why \Cref{alg:p2} is well-defined -- why the set $Q^{\ell}_t(s_t^{\ell+1},\ldots,s_t^k)$ is nonempty when we attempt to send the $\ell$'th server to a uniformly random point in it in steps~\ref{alg:mov1} and~\ref{alg:mov2} -- and why every request gets served. We provide an answer now.


\begin{lemma}\label{lem:nonempty}
For every $t\in \{1,\ldots,T\}$ the following statements hold with probability one.
\begin{enumerate}
\item For every $\ell\in \{0,\ldots,k\}$, there exists a feasible labeling $\gamma$ of $\II_t$ with respect to $\rho_t$ such that $\gamma(L^i_t)=s^i_t$ for all $i\in\{\ell+1,\ldots,k\}$.
\item For every $\ell\in \{1,\ldots,k\}$, the set $Q^{\ell}_t(s^{\ell+1}_t,\ldots,s^k_t)$ is non-empty.
\item \Cref{alg:p2} serves the $t$'th request.
\end{enumerate}
\end{lemma}

\begin{proof}
For an arbitrary $t\in\{1,\ldots,T\}$, we prove the lemma by reverse induction on $\ell\in\{0,\ldots,k\}$ in an interleaved manner. More precisely, as the base case, we prove the first claim for $\ell=k$. Assuming that the first claim holds for an arbitrary $\ell>0$, we prove that the second claim holds for the same $\ell$. Assuming that the second claim holds for an arbitrary $\ell>0$, we prove that the first claim holds for $\ell-1$. Finally, assuming that the first claim holds for $\ell=0$, we prove that the third claim holds.
    
    As the base case, we need to prove the first claim for $\ell=k$. We know that the service pattern $\II_t$ that the adversary provides is feasible. This implies that there exists a feasible labeling $\gamma$ of $\II_t$ with respect to $\rho_t$. The condition $\gamma(L^i_t)=s^i_t$ for all $i\in\{\ell+1,\ldots,k\}$ is vacuously true.
    
    For the inductive step, assume that the first claim holds for some $0<\ell\leq k$. Hence, there exists a feasible labeling $\gamma$ of $\II_t$ such that $\gamma(L^i_t)=s^i_t,$ for all $i\in\{\ell+1,\ldots,k\}$. By \Cref{feasible-labels-set}, the point $\gamma(L^\ell_t)$ lies in $Q^{\ell}_t(s^{\ell+1}_t,\ldots,s^k_t)$. Thus, $Q^{\ell}_t(s^{\ell+1}_t,\ldots,s^k_t)\neq\emptyset$.
    
    We designed the algorithm so that $s^{\ell}_t$ is guaranteed to be in $Q^{\ell}_t(s^{\ell+1}_t,\ldots,s^k_t)$. By \Cref{feasible-labels-set}, there exists a feasible labeling $\gamma'$ of $\II_t$ such that $\gamma'(L^i_t)=s^i_t$ for all $i\in\{\ell,\ldots,k\}$. Hence, the first claim holds for $\ell-1$ as well. This proves the first two claims.
    
    Finally, since the first claim holds for $\ell=0$, there exists a feasible labeling $\gamma$ of $\II_t$ with respect to $\rho_t$ such that $\gamma(L_t^i)=s_t^i$ for all $i\in\{1,\ldots,k\}$. By definition of feasibility of a labeling (\Cref{feasible-labeling}), $\gamma$ must assign the label $\sigma_t$ to some interval in $\II_t$ that contains $t$. Since the only intervals in $\II_t$ that contains $t$ are the $L_t^i$'s, we must have $\sigma_t=s_t^i$ for some $i$. Since $s_t^i$'s are the positions of the algorithm's servers after processing the $t$'th request, the request gets served.
\end{proof}

\subsection{Competitive Analysis}\label{sec:comptv}
We begin by proving that the algorithm indeed maintains \Cref{inv:random-point} after serving every request.

\begin{lemma}\label{random-point}
    For every $t\in \{1,\ldots,T\}$, every $\ell\in \{1,\ldots,k\}$, and every $p^{\ell+1},\ldots,p^k \in U$, conditioned on $s^i_t = p^i$ for all $i\in \{\ell+1,\ldots,k\}$ and $Q_t^{\ell}(p^{\ell+1},\ldots, p^k)\neq\emptyset$, $s^{\ell}_t$ is a uniformly random point in $Q_t^{\ell}(p^{\ell+1},\ldots, p^k)$.
\end{lemma}

\begin{proof}
We prove this lemma by induction of time $t$. The base case of $t=1$ is true because we are assuming $\ell_1=k$, which makes the condition in step~\ref{alg:if1} of the algorithm true.

Consider the inductive case, where $t>1$. Note that the algorithm executes exactly one step out of \ref{alg:mov1}, \ref{alg:mov2}, and \ref{alg:no-movement}. Conditioned on the algorithm executing step \ref{alg:mov1} or \ref{alg:mov2}, $s^{\ell}_t$ is located at a uniformly random point in $Q_t^{\ell}(p^{\ell+1},\ldots, p^k)$ by design. On the other hand, suppose the algorithm executes step \ref{alg:no-movement}, that is, the checks in steps \ref{alg:if1} and \ref{alg:if2} fail. Then the fact that the check of step~\ref{alg:if1} failed implies that the algorithm did not move any server heavier than the $\ell$'th. Thus $(s_{t-1}^{\ell+1},\ldots,s_{t-1}^k)=(s_t^{\ell+1},\ldots,s_t^k)=(p^{\ell+1},\ldots,p^k)$. Therefore, by induction hypothesis, $s_{t-1}^{\ell}$ is a uniformly random point in $Q^{\ell}_{t-1}(p^{\ell+1},\ldots,p^k)$. 
Additionally, conditioned on the failure of the check in step~\ref{alg:if2}, $ s^{\ell}_{t-1}\in Q^{\ell}_t(s_t^{\ell+1},\ldots,s_t^k)$, so $s^{\ell}_t=s^{\ell}_{t-1}$ is a uniformly random point in $Q^{\ell}_{t-1}(p^{\ell+1},\ldots,p^k)\cap Q^{\ell}_t(p^{\ell+1},\ldots,p^k)$. But by \Cref{subset-feasible-labels}, $Q^{\ell}_{t-1}(p^{\ell+1},\ldots,p^k)\cap Q^{\ell}_t(p^{\ell+1},\ldots,p^k)=Q^{\ell}_t(p^{\ell+1},\ldots,p^k)$, so $s^{\ell}_t$ is a uniformly random point in $Q^{\ell}_t(p^{\ell+1},\ldots,p^k)$. Thus, irrespective of which one of steps \ref{alg:mov1}, \ref{alg:mov2}, and \ref{alg:no-movement} is executed, $s^{\ell}_t$ is a uniformly random point in $Q^{\ell}_t(p^{\ell+1},\ldots,p^k)$, and this implies the claim.
\end{proof}

Having established that our algorithm is well-defined and that it indeed serves every request, we now focus on bounding the cost of algorithm's solution. For each $\ell\in\{1,\ldots,k\}$, define the random variables $X^{\ell}$ and $Y^{\ell}$ to be the number of \mone s  and \mtwo s respectively, of the algorithm's $\ell$'th server. First, we bound $X^\ell$ for all $\ell$ as follows. 

\begin{lemma}\label{lem:xl}
    The following inequalities hold with probability one.
    \begin{enumerate}
        \item $X^\ell \leq X^{\ell+1} + Y^{\ell+1} + |\II_T^\ell|$ for all $\ell\in\{1,\ldots,k-1\}$.
        \item $X^k \leq|\II_T^k|$.
    \end{enumerate}
\end{lemma}

\begin{proof}
For $\ell<k$, every \mone\ of the $\ell$'th server happens at time $t$ only if \texttt{flag} is true or $\ell\leq\ell_t$. Observe that in the former case, the $(\ell+1)$'th server must have moved at time $t$ too, so we charge the movement of the $\ell$'th server to the movement of the $(\ell+1)$'th server, which could be either forced or unforced. In the latter case, since $\II_T$ is a hierarchical service pattern, a new interval starts at time $t$ in the $\ell$'th level $\II_T^{\ell}$ of $\II_T$, so we charge the movement of the $\ell$'th server to that interval. The argument for the second claim is the same as above except that \texttt{flag} is never true; a forced movement of the $k$'th server happens at time $t$ only if $\ell_t=k$.
\end{proof}

Next, we bound $Y^{\ell}$ by first bounding the number of \mtwo s of the $\ell$'th server in an arbitrary interval in which no \mone\ of the $\ell$'th server happens.

\begin{lemma}\label{lem:hnl}
    For every $\ell\in\{1,\ldots,k\}$, every $p^{\ell+1},\ldots,p^k\in U$, and every $t_{\text{begin}}$, $t_{\text{end}}$ such that $1\leq t_{\text{begin}}<t_{\text{end}}\leq T+1$, the following holds. 
    Conditioned
    on the event that $s^j_t=p^j$ for all $j\in\{\ell+1,\ldots,k\}$, and all $t \in (t_{\text{begin}},t_{\text{end}})$, the expected number of \mtwo s of the algorithm's $\ell$'th server is at most $H(n_{\ell})$, where $H$ denotes the harmonic function defined as $H(n)=1+1/2+\cdots+1/n$.
\end{lemma}

\begin{proof}
    Conditioned on the event that $s^j_t=p^j$ for all $j\in\{\ell+1,\ldots,k\}$ and all $t \in (t_{\text{begin}},t_{\text{end}})$,
we use \Cref{subset-feasible-labels} to claim that $Q^{\ell}_{t-1}(p^{\ell+1},\ldots,p^k)\supseteq Q^{\ell}_{t}(p^{\ell+1},\ldots,p^k)$ for every $t\in(t_{\text{begin}},t_{\text{end}})$. For brevity we write $Q_t$ for $Q^{\ell}_{t}(p^{\ell+1},\ldots,p^k)$.
    Let $Z_t$ be the indicator random variable of the event that an \mtwo\ happens at time $t$. From the algorithm, we know that this event happens if and only if $s_{t-1}^{\ell}\in Q_{t-1}\setminus Q_t$. From \Cref{random-point}, we know that $s^\ell_{t-1}$ is a uniformly random point in $Q_{t-1}$. Thus,
\[\mathbb{E}[Z_t]=\frac{|Q_{t-1}|-|Q_t|}{|Q_{t-1}|} \leq \frac{1}{|Q_{t-1}|} + \frac{1}{|Q_{t-1}|-1} + \cdots + \frac{1}{|Q_{t}|+1}=H(|Q_{t-1}|)-H(|Q_{t}|)\]
Let $t_1$ denote the earliest time $t$ for which $Q_t\subsetneq Q_{t-1}$. Then the expected number of \mtwo s is bounded as
\[\sum_{t=t_1}^{t_{\text{end}}-1}\mathbb{E}[Z_t]\leq1+\sum_{t=t_1+1}^{t_{\text{end}}-1}H(|Q_{t-1}|)-H(|Q_{t}|)=1+H(|Q_{t_1}|)-H(|Q_{t_{\text{end}}-1}|)\text{.}\]
Since $Q_{t_1}\subsetneq Q_{t_1-1}\subseteq U$, by \Cref{dichotomy-theorem-extn}, we have $|Q_{t_1}|\leq n_{\ell}$. By the second claim of \Cref{lem:nonempty}, $|Q_{t_{\text{end}}-1}|\geq1$. Thus, the expected number of \mtwo s is at most $H(n_{\ell})$.
\end{proof}

\begin{lemma}\label{lem:yl}
For every $\ell\in\{1,\ldots,k\}$, we have $\mathbb{E}[Y^{\ell}]\leq H(n_{\ell})\cdot\mathbb{E}[X^{\ell}]$.
\end{lemma}

\begin{proof}
Let us condition on the sequence of timestamps at which a \mone\ of the $\ell$'th server takes place. If $t_1,t_2$ are two consecutive timestamps in this sequence, \aay{it is easy to notice that the servers $\ell+1,\ldots,k$ remain at the same position throughout this interval}. Then \Cref{lem:hnl} applied to the interval $(t_1,t_2)$ implies that the expected number of \mtwo s of the $\ell$'th server in this interval is at most $H(n_{\ell})$. Summing up over all pairs $t_1,t_2$ of consecutive timestamps, we get $\mathbb{E}[Y^{\ell}|X^{\ell}=x]\leq H(n_{\ell})\cdot x$, and therefore, $\mathbb{E}[Y^{\ell}]\leq H(n_{\ell})\cdot\mathbb{E}[X^{\ell}]$.
\end{proof}

\maintheorem*
\begin{proof}
    From \Cref{lem:nonempty}, we already know that \Cref{alg:p2} serves every request in $\rho_T$. Towards proving competitiveness of the algorithm, we first define the constants $c_k,\ldots,c_1$ inductively as follows: $c_k = H(n_{k})+1$ and $c_{\ell}=(H(n_{\ell})+1)\cdot (c_{\ell+1}+1)$, for every
    $\ell\in\{1,\ldots,k-1\}$. We claim that for every $\ell\in\{1,\ldots,k\}$, the expected number of movements of the algorithm's $\ell$'th server, which equals $\mathbb{E}[X^\ell]+\mathbb{E}[Y^\ell]$, is at most $c_{\ell}$ times the number of movements of the adversary's $\ell$'th server, which equals $|\II^\ell_T|$. 
    We prove this claim using reverse induction on $\ell$ from $k$ to $1$. 
    
    For the base case, i.e.\ $\ell=k$, from \Cref{lem:yl} we have that $\mathbb{E}[Y^k]\leq H(n_{k})\cdot\mathbb{E}[X^{k}]$. From \Cref{lem:xl}, we know that $\mathbb{E}[X^{k}]\leq|\II_T^k|$. Thus, the expected number of algorithm's $k$'th server movements is,
\[\mathbb{E}[X^k]+\mathbb{E}[Y^k]\leq(H(n_k)+1)\cdot\mathbb{E}[X^{k}]=c_k\cdot|\II_T^k|\text{.}\]
    
    For the inductive case, assume that $\mathbb{E}[X^{\ell+1}]+\mathbb{E}[Y^{\ell+1}] \leq c_{\ell+1}\cdot|\II^{\ell+1}_T|$, for an arbitrary $\ell\in\{1,\ldots,k-1\}$. From \Cref{lem:yl}, we have that $\mathbb{E}[Y^\ell]\leq H(n_{\ell})\cdot\mathbb{E}[X^{\ell}]$, and from \Cref{lem:xl}, we have that $\mathbb{E}[X^{\ell}]\leq\mathbb{E}[X^{\ell+1}]+\mathbb{E}[Y^{\ell+1}] +|\II_T^\ell|$. Thus, we have,
\[\mathbb{E}[X^\ell]+\mathbb{E}[Y^\ell]\leq(H(n_{\ell})+1)\cdot\mathbb{E}[X^{\ell}]\leq(H(n_{\ell})+1)\cdot\left(\mathbb{E}[X^{\ell+1}]+\mathbb{E}[Y^{\ell+1}]+|\II_T^\ell|\right)\text{.}\]
Recall that $\II_T=(\II_T^1,\ldots,\II_T^k)$ is a hierarchical service pattern, and therefore, $|\II_T^{\ell+1}|\leq|\II_T^{\ell}|$. Using this fact, the induction hypothesis, and the definition of $c_{\ell}$, we get,
\[\mathbb{E}[X^\ell]+\mathbb{E}[Y^\ell]\leq(H(n_{\ell})+1)\cdot\left(c_{\ell+1}\cdot|\II^{\ell+1}_T|+|\II_T^\ell|\right)\leq(H(n_{\ell})+1)\cdot(c_{\ell+1}+1)\cdot|\II_T^\ell|=c_\ell\cdot |\II_T^\ell|\text{,}\]
as required.

As a consequence of the above inductive claim, the total cost of the algorithm is at most $\max\{c_1,\ldots,c_k\}=c_1$ times the cost of the adversary's service pattern. Moreover, from the recurrence relation defining $c_k,\ldots,c_1$ and the upper bound on $n_{\ell}$ from \Cref{dichotomy-theorem}, it is clear that $c_1$ is $2^{O(k^2)}$. Thus, the competitive ratio of our algorithm is $2^{O(k^2)}$.
\end{proof}

\section{Concluding Remarks and Open Problems}

The main open question of finding the randomized competitive ratio of weighted $k$-server on uniform metrics still remains unresolved. Our decomposition approach and the randomized algorithm for \ptwo\ imply that the task of designing a $2^{\poly(k)}$-competitive randomized algorithm for weighted $k$-server on uniform metrics is equivalent to designing a $2^{\poly(k)}$-competitive algorithm for \pone. We do not know any non-trivial bounds on the competitive ratio of \pone\ and it is not even clear whether it is easier or harder than \ptwo\ in terms of competitiveness. While the known lower bound constructions for weighted $k$-server also apply to \ptwo, these constructions fail to get a lower bound on the competitive ratio of \pone. 
We therefore propose the open problem of finding bounds on the competitive ratio of \pone\ in the deterministic as well as randomized setting. 

In the deterministic setting, the competitive ratio of \pone\ is bounded from below by the competitive ratio of weighted $k$-server divided by the competitive ratio of \ptwo, again due to~\Cref{two-problem-theorem}. However, for this to give a non-trivial lower bound on the competitive ratio of \pone, we require an upper bound on the competitive ratio of \ptwo\ that is less than the known lower bound on the competitive ratio of weighted $k$-server. Unfortunately, no such upper bound is known. Thus, showing a separation between weighted $k$-server and \ptwo\ is an interesting open problem.

Additionally, we also believe that closing the quadratic gap between the exponents in the upper and lower bounds on the randomized competitive ratio of \ptwo\ is an interesting open problem, because it will result in a better understanding of the weighted $k$-server problem.

Finally, for the weighted $k$-server problem with $k>2$,
no weight-independent and metric-independent upper bounds on the competitive ratio are known on any well-structured class of metrics larger than the class of uniform metrics. Proving such bounds seems rather ambitious, given our limited understanding of weighted $k$-server on uniform metrics. 



\bibliographystyle{plainurl}
\bibliography{arxiv}

\appendix

\section{Lower Bound for \texorpdfstring{\ptwo}{(LaTeX error)}}\label{sec_lowerbound}

In this section, we show that the lower-bound construction for \WKS\ by Ayyadevara and Chiplunkar~\cite{AyyadevaraC_ESA21} applies to \ptwo\ and gives the same lower bound. In~\cite{AyyadevaraC_ESA21} the generation of an adversarial request sequence involves repeatedly calling a randomized recursive procedure named \textsf{strategy}.
Their analysis bounds the adversary's cost and the expected cost of an arbitrary deterministic algorithm, both amortized per \textsf{strategy} call. Then the exponential lower bound is established by an application of Yao's principle.
We show that essentially the same adversarial construction of input distribution works for \ptwo. The only change needed in the construction is that now the adversary must provide its service pattern along with the requested point in an online manner. 

The weights of the servers are $1, \beta, \ldots, \beta^{k-1}$ where $\beta$ is a large integer. The sequence $n_0, n_1, \ldots$ is defined as $n_0 = 1$ and for $\ell>0$,
\[n_{\ell}=\left(\left\lceil\frac{n_{\ell-1}}{2}\right\rceil+1\right)\cdot\left(\left\lfloor\frac{n_{\ell-1}}{2}\right\rfloor+1\right)\text{.}\]
The adversarial strategy in~\cite{AyyadevaraC_ESA21} uses the following combinatorial result.

\begin{fact}[Bansal et al.~\cite{BansalEK_FOCS17}]\label{lem_setsystem}
Let $\ell\in\mathbb{N}$ and let $P$ be a set of $n_{\ell}$ points. There exists a set-system $\mathcal{Q}_{\ell}\subseteq2^P$ satisfying the following properties.
\begin{enumerate}
\item $\mathcal{Q}_{\ell}$ contains $\lceil n_{\ell-1}/2\rceil+1$ sets, each of size $n_{\ell-1}$.
\item For every $p\in P$, there exists a set in $\mathcal{Q}_{\ell}$ not containing $p$.
\item For every $p\in P$, there exists a $q\in P$ such that every set in $\mathcal{Q}_{\ell}$ contains at least one of $p$ and $q$.
\end{enumerate}
\end{fact}

We modify the adversarial strategy from~\cite{AyyadevaraC_ESA21} for weighted $k$-server to get the following strategy for \ptwo.

\begin{algorithm}[H]
\SetAlgorithmName{Procedure}
\renewcommand{\thealgorithm}{}
\caption{\textsf{adversary}}\label{adversary}
Mark all points in $S$\;
\RepTimes{infinitely many}{
Pick a point $p$ uniformly at random from $S$ (with replacement)\;
Mark $p$\;
\If{All points in $S$ are marked}{\label{all-pages-marked}
$\ell_{ext}\gets k$\;
Unmark all points $q\in S$ other than $p$\;
}
\Else{
$\ell_{ext} \gets k-1$\;
}
Call \textsf{strategy}$(k-1,S\setminus\{p\}, \ell_{ext})$\;
}
\end{algorithm}

\begin{algorithm}[H]
\SetAlgorithmName{Procedure}
\renewcommand{\thealgorithm}{}
\caption{\textsf{strategy}$(\ell,P,\ell_{ext})$ (Promise: $|P|=n_{\ell}$ and $\ell_{ext}\geq\ell$.)}
\eIf{$\ell=0$ (and therefore, $|P|=n_0=1$)}{
Output $(p,\ell_{ext})$, where $p$ is the unique point in $P$\;
}{
Construct the set-system $\mathcal{Q}_{\ell}\subseteq2^P$ using~\Cref{lem_setsystem}\;
\RepTimes{$(\beta-1)\cdot\left(\lceil n_{\ell-1}/2\rceil+1\right)$}{
Pick a set $P'$ uniformly at random from $\mathcal{Q}_{\ell}$ (with replacement)\;
Call \textsf{strategy}$(\ell-1,P',\ell_{ext})$\;
$\ell_{ext}\gets\ell-1$\;
}
}
\end{algorithm}

The following claim bounds the expected cost of an arbitrary online algorithm for every \textsf{strategy} call made by \textsf{adversary}. Its proof is identical to the proof of Corollary 6 in~\cite{AyyadevaraC_ESA21}. It is noteworthy that all the arguments involved in that proof go through even in the revealed service pattern setting.

\begin{lemma}\label{algcost}
    For $\ell_{ext} = k $ or $k-1$, the expected cost of the algorithm per \textsf{strategy}$(k-1, P, \ell_{ext})$ call made by \textsf{adversary} is $(\beta-1)^{k-1}/(n_{k-1}+1)$.
\end{lemma}

We now show that the service pattern $\II=(\II^1,\ldots,\II^k)$ created by the procedure \textsf{adversary} is feasible, that is, it can be labeled in a way that all requests get served. We start by noting the following.

\begin{observation}\label{obs-service-pattern}
Let $(p_{t_1},\ell_{t_1}),\ldots,(p_{t_2},\ell_{t_2})$ be the input generated by a \textsf{strategy}$(\ell, P, \ell_{ext})$ call which starts at time $t_1$ and ends at time $t_2$. Then $\ell_{t_1}=\ell_{ext}\geq\ell$ and $\ell_{t_1+1},\ldots,\ell_{t_2}$ are all less than $\ell$. As a consequence, the following statements about the adversary's service pattern $\II=(\II^1,\ldots,\II^k)$ hold.
\begin{enumerate}
\item For all $\ell'\geq\ell$, a single interval in $\II^{\ell'}$ covers the interval $[t_1,t_2+1)$.
\item For all $\ell'\leq\ell$, the interval in $\II^{\ell'}$ covering $t_1$ starts at $t_1$, and the interval in $\II^{\ell'}$ covering $t_2$ ends at $t_2+1$.
\end{enumerate}
In particular, $[t_1,t_2+1)$ is an interval in $\II^{\ell}$.
\end{observation}


\begin{lemma}\label{feasible-induction}
Consider an arbitrary \textsf{strategy}$(\ell, P, \ell_{ext})$ call which starts at time $t_1$, ends at time $t_2$, and generates the input $(p_{t_1},\ell_{t_1}),\ldots,(p_{t_2},\ell_{t_2})$. Suppose for all $\ell'>\ell$, the single interval in $\II^{\ell'}$ covering $[t_1,t_2+1)$ is labeled with some point $p_{\ell'}$, such that $P\cap\{p_{\ell+1},\ldots,p_{k}\}\neq\emptyset$. Then the intervals in $\II^1,\ldots,\II^{\ell}$ that intersect $[t_1,t_2+1)$ (and therefore, are subsets of $[t_1,t_2+1)$) can be labeled in such a way that for all $t\in\{t_1,\ldots,t_2\}$ there exists $i\in\{1,\ldots,k\}$ such that the unique interval in $\II^i$ covering $t$ is labeled with $p_t$.
\end{lemma}

\begin{proof}
We prove by induction on $\ell$. For $\ell=0$, the set $P$ contains a single point, which gets requested. Since $P\cap \{p_1, \ldots, p_k\}\neq \emptyset$, the claim holds. 

For $\ell>0$, consider a point $p\in P\cap \{p_{\ell+1},\ldots, p_k\}$. From the third property in~\Cref{lem_setsystem}, there exists a point $q\in P$ such that every set in the set system $\mathcal{Q}_\ell$ contains at least one of the points $p$ or $q$. Label the interval $[t_1, t_2+1)$ in $\mathcal{I}^\ell$ by such a point $q$. Consider an arbitrary recursive call \textsf{strategy}$(\ell-1, P', \ell_{ext})$ which starts at time $t'_1\geq t_1$ and ends at time $t'_2\leq t_2$. We have $P'\cap \{q, p_{\ell+1}, \ldots, p_k\}\neq \emptyset$. By induction hypothesis, the intervals in $\mathcal{I}^1, \ldots, \mathcal{I}^{\ell-1}$ that intersect $[t'_1, t'_2+1)$ can be labeled in such a way that for all $t\in\{t'_1,\ldots,t'_2\}$ there exists $i\in\{1,\ldots,k\}$ such that the unique interval in $\II^i$ covering $t$ is labeled with $p_t$.
\end{proof}

\begin{lemma}
The service pattern $\II=(\II^1,\ldots,\II^k)$ created by the procedure \textsf{adversary} is feasible.
\end{lemma}

\begin{proof}
Every interval in $\II^k$ starts exactly when all points in $S$ are found to be marked. For every interval $I^k$ in $\II^k$ do the following. Label it by the point $q$ whose marking results in the beginning of the next interval in $\II^k$. In other words, $q$ is the last point to get marked after the unmarking step in the beginning of $I^k$. Thus, $q$ is never sampled by \textsf{adversary} during the interval $I^k$, and therefore $q$ belongs to the set $S\setminus\{p\}$ passed to every \textsf{strategy} call made by \textsf{adversary} during the interval $I^k$. Thus, by \Cref{feasible-induction} for $\ell=k-1$, the intervals in $\II^1,\ldots,\II^{k-1}$ that are subsets of $I^k$ can be labeled in such a way that all requests given during the interval $I^k$ are served. Thus, the service pattern $\II$ is feasible.
\end{proof}

Now, we bound the cost of the service pattern $\II=(\II^1,\ldots,\II^k)$ created by the procedure \textsf{adversary}. 
We start by bounding the total cost of intervals created during a \textsf{strategy}$(\ell, P, \ell_{ext})$ call.

\begin{lemma}\label{advcostrec}
    Define the sequence $c_0, c_1, \ldots$ inductively as follows: $c_0=0$ and for $\ell>0$
    \[c_\ell = \beta^{\ell-1} + \beta\cdot\left(\left\lceil n_{\ell-1}/2\right\rceil+1\right)\cdot c_{\ell-1}\]
    For an arbitrary $\ell\in \{0, \ldots, k-1\}$ and $\ell_{ext} \geq \ell$, consider a call of \textsf{strategy}$(\ell, P, \ell_{ext})$ which starts at time $t_1$ and ends at time $t_2$. The total cost of the intervals in layers $\II_1,\ldots,\II_{\ell}$ that intersect the interval $[t_1, t_2+1)$ (equivalently, are subsets of $[t_1, t_2+1)$, by~\Cref{obs-service-pattern}) is at most $c_{\ell}$.
\end{lemma}
\begin{proof}
    We prove this by induction on $\ell$. The claim is trivially true for $\ell=0$.
    For $\ell>0$, the \textsf{strategy}$(\ell, P, \ell_{ext})$ makes $(\beta-1)\cdot(\lceil n_{\ell-1}/2 \rceil+1)$ recursive calls of \textsf{strategy}$(\ell-1, P', \ell_{ext})$. For each of these calls, the following holds by induction hypothesis. If the call starts at time $t'_1\geq t_1$ and ends at time $t'_2\leq t_2$, then the total cost of intervals in layers $\II_1,\ldots,\II_{\ell-1}$ that intersect the interval $[t'_1,t'_2+1)$ is at most $c_{\ell-1}$. Since there are $(\beta-1)\cdot(\lceil n_{\ell-1}/2 \rceil+1)$ such recursive calls the total cost of intervals in layers $\II_1,\ldots,\II_{\ell-1}$ that intersect the interval $[t_1,t_2+1)$ is at most $(\beta-1)\cdot(\lceil n_{\ell-1}/2 \rceil+1) \cdot c_{\ell-1}$. Adding to this the cost $\beta^{\ell-1}$ of the interval $[t_1, t_2+1)\in\II_{\ell}$ gives the required bound.
\end{proof}


\lowerbound*

\begin{proof}
    Consider the service pattern $\mathcal{I} = (\mathcal{I}^1, \ldots, \mathcal{I}^k)$ of the adversary. Every interval in $\mathcal{I}^k$ starts with one marked point and ends just before all the points in the set $S$ are marked in the procedure \textsf{adversary}. Using the standard coupon collector argument, we get that the expected number of \textsf{strategy}$(k-1, S\setminus\{p\}, \ell_{ext})$ made during an interval in $\mathcal{I}^k$ is $(n_{k-1}+1)H(n_{k-1})$. Thus, the amortized cost of intervals in $\mathcal{I}^k$ per \textsf{strategy}$(k-1, S\setminus\{p\}, \ell_{ext})$ call is $\beta^{k-1}/((n_{k-1}+1)H(n_{k-1}))$.
    From~\Cref{advcostrec}, we get that the total cost of intervals in $\mathcal{I}^1, \ldots, \mathcal{I}^{k-1}$ per \textsf{strategy}$(k-1, S\setminus\{p\}, \ell_{ext})$ call is at most $c_{k-1}$. The cost of the revealed service pattern per \textsf{strategy}$(k-1, S\setminus\{p\}, \ell_{ext})$ call is at most $\beta^{k-1}/((n_{k-1}+1)H(n_{k-1})) + c_{k-1}$. The rest of the proof is identical to the proof of Theorem 2 in~\cite{AyyadevaraC_ESA21}. Essentially, for a large $\beta$, the dominant term in the adversary's cost is $\beta^{k-1}/((n_{k-1}+1)H(n_{k-1}))$, while the algorithm's cost is $\beta^{k-1}/(n_{k-1}+1)$ modulo a lower order term due to~\Cref{algcost}, thus implying a lower bound of $H(n_{k-1})=\Omega(2^k)$ on the competitive ratio.
\end{proof}

\end{document}